\documentclass{article}
\usepackage{hyperref}

\usepackage{graphicx,amssymb,amsmath}
\usepackage{amsfonts,amsthm}
\usepackage{color}
\usepackage{subfigure}
\usepackage{bm}
\usepackage{a4wide}

\graphicspath{{figs/}}

\newcommand{\Poly}{\ensuremath{\mathcal{P}}}

\newcommand{\MinDT}{\ensuremath{\mathrm{minDT}}}
\newcommand{\MaxDT}{\ensuremath{\mathrm{maxDT}}}
\newcommand{\triang}{\ensuremath{T}}

\newtheorem{theorem}{Theorem}[section]
\newtheorem{prop}[theorem]{Proposition}
\newtheorem{cor}[theorem]{Corollary}

\newtheorem{lemma}[theorem]{Lemma}

\title{The Dual Diameter of Triangulations$^*$} %

\author{
Matias~Korman\thanks{National
Institute of Informatics (NII), Tokyo, Japan. {\tt korman@nii.ac.jp}}
$^,$\thanks{JST, ERATO, Kawarabayashi Large Graph Project.}
\and Stefan~Langerman\thanks{Universit\'e Libre de
Bruxelles (ULB), Brussels, Belgium.
{\tt stefan.langerman@ulb.ac.be}. Directeur de Recherches du FRS-FNRS.}
\and Wolfgang~Mulzer\thanks{Institut f\"ur Informatik, Freie
	Universit\"at Berlin,
 Germany. {\tt mulzer@inf.fu-berlin.de}}
\and Alexander~Pilz\thanks{Institute for Software Technology,
Graz University of Technology, Austria, {\tt [apilz|bvogt]@ist.tugraz.at}}
\and Maria~Saumell\thanks{Department of Mathematics and European Centre of Excellence NTIS (New Technologies for the Information Society), University of West Bohemia, Czech Republic, {\tt saumell@kma.zcu.cz}}
\and Birgit~Vogtenhuber\footnotemark[6]
}

\begin{document}
\maketitle

\begin{abstract}
Let $\Poly$ be a simple polygon with $n$ vertices.
The \emph{dual graph} $\triang^*$ of a triangulation~$\triang$ of~$\Poly$
is the graph whose vertices correspond to the bounded faces of
$\triang$ and whose edges connect those faces of~$\triang$ that share an edge.
We consider triangulations of~$\Poly$ that minimize or maximize the 
diameter of their dual graph. 
We show that both triangulations can be constructed in $O(n^3\log n)$ time
using dynamic programming.
If $\Poly$ is convex, we show that any minimizing triangulation has dual 
diameter exactly $2\cdot\lceil\log_2(n/3)\rceil$ or 
$2\cdot\lceil\log_2(n/3)\rceil -1$, depending on~$n$. Trivially, in this
case any maximizing triangulation has dual diameter $n-2$.
Furthermore, we investigate the relationship between the dual diameter 
and the number of \emph{ears} (triangles with exactly two edges incident 
to the boundary of $\Poly$) in a triangulation. 
For convex $\Poly$, we show that there is always a triangulation that
simultaneously minimizes the dual diameter and maximizes the number of 
ears.
In contrast, we give examples of general simple polygons where every 
triangulation that maximizes the number of ears has dual diameter that 
is quadratic in the minimum possible value.
We also consider the case of point sets in general position in the plane. 
We show that for any such set of $n$ points there are triangulations with 
dual diameter in~$O(\log n)$ and in~$\Omega(\sqrt n)$.
\end{abstract}




\section*{Foreword}
Research on this topic was initiated at the \emph{Brussels Spring 
Workshop on Discrete and Computational Geometry}, which took place 
May 20--24, 2013.  The authors would like to thank all the participants in general and  Ferran Hurtado in particular. Ferran participated in the early stages of the discussion, but modestly decided not to be an author of this paper. To us he has been a teacher, supervisor, advisor, mentor, colleague, coauthor, and above all: a friend. We are very grateful that he was a part of our lives.

\section{Introduction}
Let $\Poly$ be a simple polygon with $n>3$ vertices. 
We regard $\Poly$ as a closed two-dimensional subset of the plane, 
containing its boundary. A \emph{triangulation} $\triang$ of~$\Poly$ is a 
maximal crossing-free geometric (i.e., straight-line) graph whose 
vertices are the vertices of~$\Poly$ and whose edges lie inside~$\Poly$. 
Hence, $\triang$ is an outerplanar graph. Similarly, for a set $S$ of 
$n$ points in the plane, a \emph{triangulation} $\triang$ of~$S$ is a maximal 
crossing-free geometric graph whose vertices are exactly the points of~$S$.
It is well known that in both cases all bounded faces of $\triang$ are triangles.
The \emph{dual graph} $\triang^*$ of $\triang$ is the graph 
with a vertex for each bounded face of $\triang$ and an edge between 
two vertices if and only if the corresponding triangles share an edge in 
$\triang$. If all vertices of $\triang$ are incident to the unbounded face,
then $\triang^*$ is a tree. An \emph{ear} in a 
triangulation of a simple polygon is a triangle whose vertex in the dual graph is a leaf (equivalently, two out of its three edges are edges of $\Poly$).
We call the diameter of the dual graph $\triang^*$ the \emph{dual diameter 
(of the triangulation $\triang$)}.
In the following, we will study combinatorial and algorithmic properties
of
\emph{minimum} and \emph{maximum dual diameter triangulations} 
for simple polygons and for planar point sets
(\MinDT{}s and \MaxDT{}s for short).
Note that both triangulations need not to be unique.

\paragraph{Previous Work}
Shermer~\cite{s-cbtt-91} considers \emph{thin} and \emph{bushy} 
triangulations of simple polygons, i.e., triangulations that minimize or 
maximize the number of ears.
He presents algorithms for computing a thin triangulation in time
$O(n^3)$ and a bushy triangulation in time $O(n)$.
Shermer also claims that bushy triangulations are useful for finding paths 
in the dual graph, as is needed, e.g., in geodesic algorithms.
In that setting, however, the running time is not actually determined 
by the number of ears, but by the dual diameter of the triangulation.
Thus, bushy triangulations are only useful for geodesic problems if 
there is a connection between maximizing the number of ears and minimizing
the dual diameter.
While this holds for convex polygons, we show that, in general, there 
exist polygons for which no $\MinDT$ maximizes the number of ears.
Moreover, we give examples where forcing a single ear into a 
triangulation may almost double the dual diameter, and the dual diameter 
of any bushy triangulation may be quadratic in the dual diameter of a 
$\MinDT$.

The dual diameter also plays a role in the study of edge flips: 
given a triangulation $T$, an \emph{edge flip} is the operation of replacing
a single edge of $T$ with another one so that the resulting graph
is again a valid triangulation.
In the case of convex polygons, edge flips correspond to rotations in the
dual binary tree~\cite{sleator_tarjan_thurston}.
For this case, Hurtado, Noy, and 
Urrutia~\cite{hurtado_noy_urrutia, urrutia_flip_talk} show that a 
triangulation with dual diameter $k$ can be transformed into a fan 
triangulation by a sequence of most $k$ parallel flips (i.e., two 
edges not incident to a common  triangle may be flipped simultaneously).
They also obtain a triangulation with logarithmic dual diameter 
by recursively cutting off a linear number of ears.

While we focus on the dual graph of a triangulation, distance
problems in the primal graph have also been considered.
For example, Kozma~\cite{kozma} addresses the problem of finding a 
triangulation 
that minimizes the total link distance over all vertex pairs.
For simple polygons, he gives a sophisticated $O(n^{11})$ time 
dynamic programming algorithm.
Moreover, he shows that the problem is strongly NP-complete for 
general point sets when arbitrary edge weights are allowed
and the length of a path is measured as the sum of the
weights of its edges.

\paragraph{Our Results}
In Section~\ref{sec_ears}, we present several properties of the dual 
diameter for triangulations of simple polygons.
Among other results, we calculate the exact dual diameter of $\MinDT$s and 
$\MaxDT$s of convex polygons, which
can be obtained by 
maximizing and minimizing the number of ears of the 
triangulation, respectively.
On the other hand, we show that there exist simple polygons where 
the dual diameter of any $\MinDT$ is $O(\sqrt{n})$, while that 
of any triangulation that maximizes the number of ears is in $\Omega(n)$.
Likewise, there exist simple polygons where the dual diameter of any 
triangulation that minimizes the number of ears is in $O(\sqrt{n})$, 
while the maximum dual diameter is linear.
In Section~\ref{sec_poly}, we present efficient algorithms to construct a 
$\MinDT$ and a $\MaxDT$ for any given simple polygon.

Finally, in Section~\ref{sec_points} we consider the case of planar
point sets, showing that for any point set in the plane in general position there
are triangulations with dual diameter in
$O(\log n)$ and in~$\Omega(\sqrt{n})$, respectively.

\section{The Number of Ears and the Diameter}\label{sec_ears}

The dual graph of any triangulation $\triang$ has maximum degree $3$. 
In this case,
the so-called \emph{Moore bound} implies that the dual diameter of 
$\triang$ is at least $\log_2(\frac{t+2}{3})$, where $t$ is the number of 
triangles in $\triang$ (see, e.g.,~\cite{moore_survey}).
For convex polygons, we can compute the minimum dual diameter exactly.

\begin{prop}\label{prop_convex}
Let $\Poly$ be a convex polygon with $n \geq 3$ vertices, and let
$m \geq 1$ such that
$n \in \{3\cdot 2^{m-1} + 1, \dots, 3 \cdot 2^m\}$. 
Then any 
\MinDT\ of $\Poly$  has dual diameter
  $2\cdot\lceil\log_2(n/3)\rceil -1$ if
  $n \in \{3\cdot 2^{m-1} + 1, \dots, 4 \cdot 2^{m-1}\}$,
  and
  $2\cdot\lceil\log_2(n/3)\rceil$ if
  $n \in \{4 \cdot 2^{m-1} + 1, \dots, 3\cdot 2^m\}$, for
  some $m \geq 1$.
\end{prop}
\begin{proof} 
The dual graph of any triangulation of~$\Poly$ is a tree with $n-2$ 
vertices and maximum degree $3$; see \figurename~\ref{fig_convex_tree} for 
an example.
Furthermore, every tree with $n-2$ vertices and maximum degree $3$ 
is dual to some triangulation of~$\Poly$.

For the upper bound, suppose first that $n = 3\cdot 2^m$, for some 
$m \geq 1$. We define a triangulation $\triang_1$ as follows.
It has a central triangle that splits $\Poly$ into three 
sub-polygons, each with $2^m$ edges on the boundary. 
For each sub-polygon, the dual tree for $\triang_1$ 
is a full binary tree of height 
$m-1$ with $2^{m-1}$ leaves; see 
\figurename~\ref{fig_convex_full}.
The leaves of~$\triang_1^*$ correspond to the ears of~$\triang_1$.
The shortest path between any two ears in two different 
sub-polygons has length exactly
$2(m-1) + 2 = 2 \log_2(n/3)$. The shortest path between any two ears in 
the same sub-polygon has length at most $2(m-1)$.
Thus, the dual diameter of~$\triang_1$ is $2 \log_2(n/3)$.

Now let $n \in \{3 \cdot \frac{4}{3} \cdot 2^{m-1} + 1, \dots,  
3\cdot 2^m-1\},$ and consider the triangulation $\triang_2$ of~$\Poly$ 
obtained by cutting off $3\cdot 2^m -n \leq 2 \cdot 2^{m-1} - 1$ 
ears that are consecutive in the convex hull from $\triang_1$. Then $\triang_2^*$ is a subtree
of~$\triang_1^*$. Since $\triang_1$ has $3\cdot 2^{m-1}$ 
ears, $\triang_2$ has at least $2^{m-1}+1$ ears in common with $\triang_1$.
Two of them lie in different sub-polygons, so the dual diameter remains 
$2m = 2\cdot\lceil\log_2(n/3)\rceil$.

Finally, for 
$n \in \{3\cdot 2^{m-1} + 1, \dots, 3 \cdot \frac{4}{3} \cdot 2^{m-1}\}$, 
if we remove  $3\cdot 2^m - n \leq  3\cdot2^{m-1}-1$ ears from $\triang_1$
such that all ears in two of the sub-polygons are removed, we get
a triangulation with dual 
diameter $2m-1=2\cdot\lceil\log_2(n/3)\rceil-1$; see 
\figurename~\ref{fig_convex_sparse}.
	
\begin{figure}[htb]
\centering
\subfigure[]{\label{fig_convex_tree}\includegraphics[scale=0.55, page=3]{fig_convex}} \hspace{0.3cm}
\subfigure[]{\label{fig_convex_full}\includegraphics[scale=0.55, page=1]{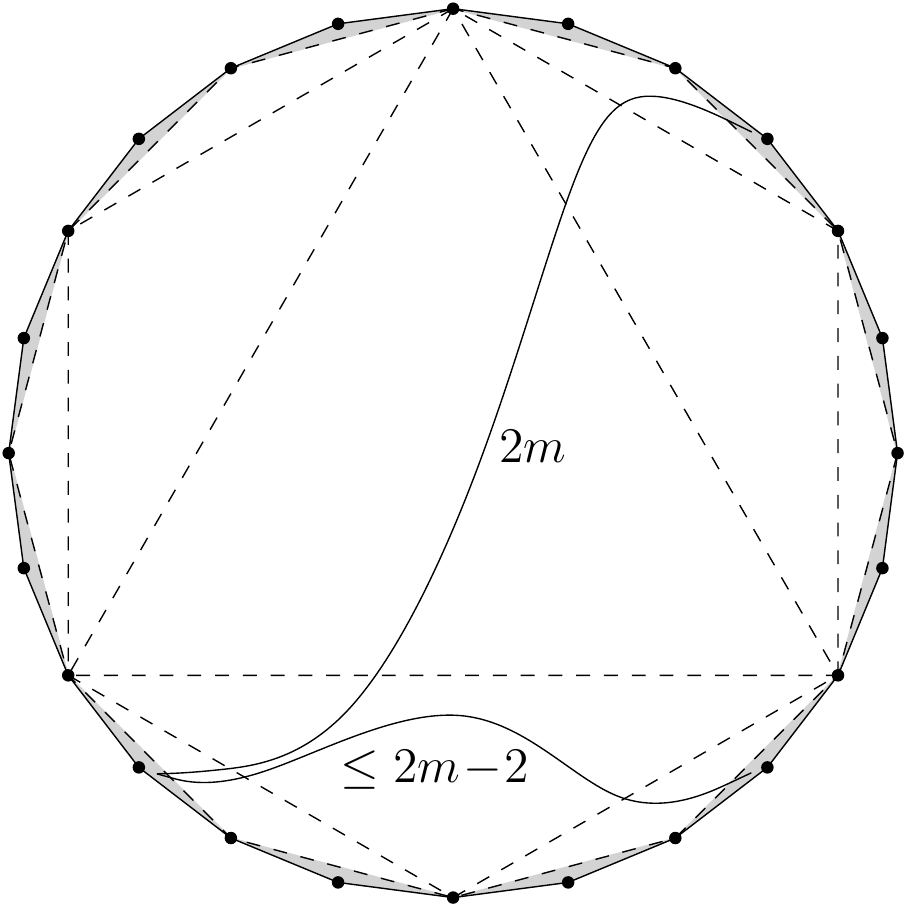}} \hspace{0.3cm}
\subfigure[]{\label{fig_convex_sparse}\includegraphics[scale=0.55, page=2]{fig_convex}}
\caption{The convex case. 
	(a) A triangulation and its dual tree. The ears 
	are gray. 
	(b) 
	The triangulation $T_1$ for 
	$m=3$. The 
	central triangle creates sub-polygons with $2^m$ edges of $\Poly$ each.
	Any path between ears in different sub-polygons has length $2m$.
	Other paths are shorter.
	(c) 
	The triangulation $T_2$ for $4\cdot 2^{m-1}$ vertices 
	($m=3$). The central 
	triangle creates three sub-polygons, one with
	$2^{m}$ edges of $\Poly$ and two with  $2^{m-2}$ edges of $\Poly$.
}
\label{fig_convex}
\end{figure}

For the lower bound, assume there is a tree $\triang^*$ with 
$n-2$ vertices, maximum degree $3$, and diameter $k$ strictly smaller 
than in the proposition. 
Consider a longest path $\pi$ in $\triang^*$ and a vertex $v$ on 
$\pi$ for which 
the distances to the endpoints of $\pi$ differ by at most one. 
By adding vertices, we can turn $\triang^*$  into a tree with 
$n' -2 > n-2$ vertices, diameter $k$, and the same structure as 
$\triang_1^*$ or $\triang_2^*$ for a convex polygon with $n'$
vertices
(with $v$ as central vertex).  
Since the upper 
bound on the dual diameter grows monotonically, this means that
the triangulation $T_1$ or $T_2$ for a convex polygon with $n$
vertices has diameter $k$, a contradiction.
\end{proof}

As the dual graph of a triangulation of \emph{any} simple polygon has 
maximum degree $3$, the proof of Proposition~\ref{prop_convex}
yields the following corollary.
\begin{cor}\label{cor_lb}
Let $\Poly$ be a simple polygon with $n \geq 3$ vertices, and let
$m \geq 1$ such that
$n \in \{3\cdot 2^{m-1} + 1, \dots, 3 \cdot 2^m\}$. 
The dual diameter of 
any triangulation of $\Poly$ is at least 
$2\cdot\lceil\log_2(n/3)\rceil -1$ if
$n \in \{3\cdot 2^{m-1} + 1, \dots, 4 \cdot 2^{m-1}\}$,
and 
$2\cdot\lceil\log_2(n/3)\rceil$ if 
$n \in \{4 \cdot 2^{m-1} + 1, \dots, 3\cdot 2^m\}$. 
\end{cor}

Proposition~\ref{prop_convex} also shows that if $\Poly$ is convex, 
there exists a $\MinDT$ with a maximum number of ears. 
Next, we show that this does not hold for general simple polygons.
Hence, any approach that tries to construct $\MinDT$s by maximizing the 
number of ears is doomed to fail.

\begin{prop}\label{prop_leaves}
For arbitrarily large $n$, there
exist simple polygons with $n$ vertices in which any $\MinDT$ minimizes the 
number of ears. 
\end{prop}
\begin{proof}
Let $k\geq 1$ and consider the polygon $\Poly$ with $n = 4k + 8$ vertices 
sketched in \figurename~\ref{fig_ears1}.
Any triangulation of $\Poly$ has either $4$ or $5$ ears.
The triangulation in \figurename~\ref{fig_ears1_5} is the only 
triangulation with 5 ears, and it has dual diameter $4k+2$.
However, as depicted in \figurename~\ref{fig_ears1_4good}, omitting the 
large ear at the bottom allows a triangulation with $4$ ears and dual 
diameter $2k+3$.
Thus, forcing even one additional ear may nearly double the dual diameter. 
\end{proof}

Figure~\ref{fig_ears1_4bad} shows a triangulation
of $\Poly$ with $4$ ears and almost twice the 
diameter as in \figurename~\ref{fig_ears1_4good}. Thus, neither for 
minimizing the diameter nor for maximizing the number of ears this 
triangulation is desirable. 
However, it has the nice property that the 
two top ears are connected by a dual path with four interior vertices.
Below, this property will be useful when making a larger construction.

\begin{figure}[htb]
\centering
\subfigure[5 ears, dual diameter $4k+2$.]{\label{fig_ears1_5}\includegraphics[scale=0.5, page=1]{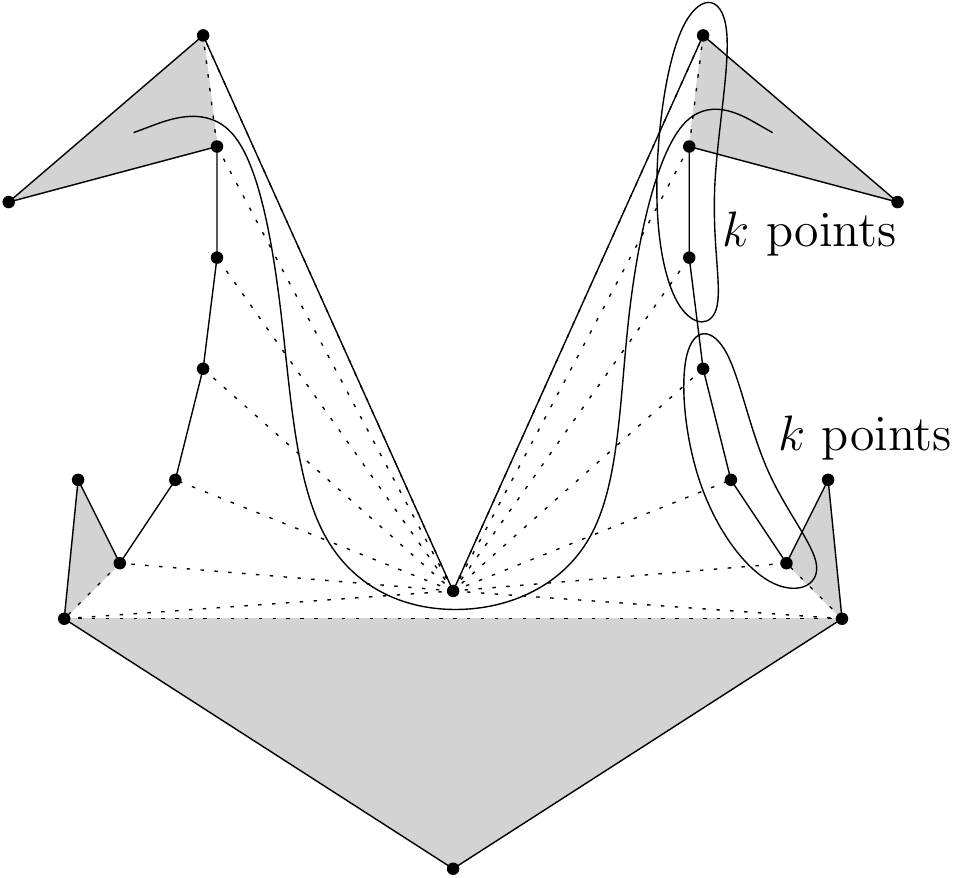}}\hspace{0.4cm}
\subfigure[4 ears, dual diameter $2k+3$]{\label{fig_ears1_4good}\includegraphics[scale=0.5, page=2]{fig_ears_new}}\hspace{0.4cm}
\subfigure[4 ears, dual diameter $4k+3$]{\label{fig_ears1_4bad}\includegraphics[scale=0.5, page=3]{fig_ears_new}}
\caption{Three triangulations of a polygon with $n= 4k + 8$ vertices 
  $(k=3)$ and paths that define their dual diameters.
  The ears are shaded.}
\label{fig_ears1}
\end{figure}

\begin{theorem}\label{thm_leaves}
For arbitrarily large $n$, there 
is a simple polygon with $n$ vertices that has minimum dual diameter
$O(\sqrt{n})$ while any triangulation
that maximizes the number of 
ears has dual diameter $\Omega(n)$. 
\end{theorem}
\begin{proof}
Let $c$ be a parameter to be determined later, and
let $\Poly'$ be the polygon constructed in Proposition~\ref{prop_leaves}. 
We construct a polygon $\Poly$ by concatenating $c$ copies of $\Poly'$ 
as in \figurename~\ref{fig_ears2}.
$\Poly$ has $n=c(4k+4)+4$ vertices.
Using the triangulation from \figurename~\ref{fig_ears1_5} for each copy,
we obtain a triangulation with the maximum number $3c+2$ of ears and dual 
diameter $c(4k+1)+1$ (the curved line in \figurename~\ref{fig_ears2} 
indicates a longest path).
On the other hand, using the triangulation from 
\figurename~\ref{fig_ears1_4good} for the leftmost and rightmost part of 
the polygon and the one from \figurename~\ref{fig_ears1_4bad} for all 
intermediate parts yields a triangulation with dual diameter $4c+4k-3$ 
that has only $2c+2$ ears.
For $c=k$, we obtain $c,k=\Theta(\sqrt{n})$.
Thus, the dual diameter for the triangulation with maximum number of ears 
is $\Theta(n)$, while the optimal dual diameter is $O(\sqrt{n})$.
\end{proof}

\begin{figure}[htb]
\centering
\subfigure[$3c+2$ ears, dual diameter $c(4k + 1) + 1$.] {\includegraphics[width=.45\columnwidth, page=4]{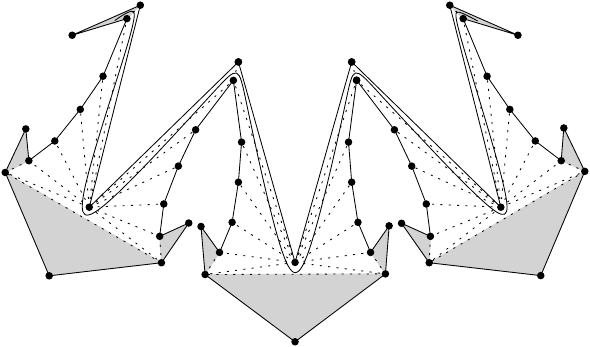}} \hspace{0.5cm} 
\subfigure[$2c+2$ ears, dual diameter $4c+4k-3$.] {\includegraphics[width=.45\columnwidth, page=5]{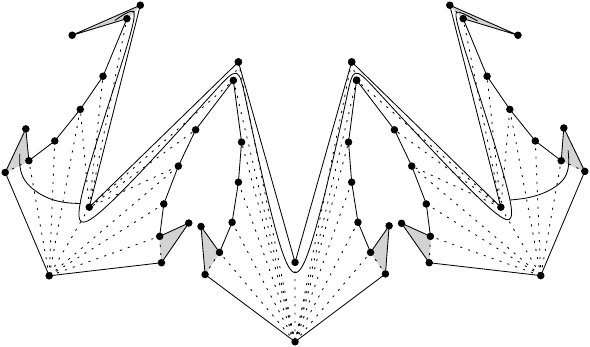}}
\caption{Two triangulations of a polygon with $n=c(4k+4)+4$ vertices $(c=k=3)$ 
and corresponding longest paths. 
The ears are shaded.}
\label{fig_ears2}
\end{figure}

Similarly, for maximizing the dual diameter, we can give examples 
where the dual diameter is suboptimal when the number of ears is minimized.

\begin{theorem}
For arbitrarily large $n$, there 
is a simple polygon with $n$ vertices that has maximum dual diameter
$\Omega(n)$ while any triangulation
that minimizes the number of 
ears has dual diameter $O(\sqrt{n})$. 
\end{theorem}
\begin{proof}
\figurename~\ref{fig:fig_ears_max} shows a triangulation of a part 
of a simple polygon.
We suppose that the indicated dual path~$\pi$ is the only one of maximum 
length.
In addition to the ears at the endpoints of~$\pi$, there are two ears at 
the vertices $v_p$ and $v_q$.
If we want to have at most one ear in this part of the polygon,
at least one of $v_p$ and $v_q$ must be connected to a 
non-neighboring vertex by a triangulation edge.
For this, the only possibilities are $v_p v_a$ and $v_q v_a$.
But then there cannot be any edge between the bottommost vertex~$v_b$ 
and the $k$ vertices between $v_q$ and $v_a$.
In particular, that part must be triangulated as shown to the right 
of \figurename~\ref{fig:fig_ears_max}.
Here, there is only one ear, but the dual diameter is reduced 
by~$k$ (assuming the remainder of the polygon is large enough).
As in the proof of Theorem~\ref{thm_leaves}, we concatenate 
$\Theta(\sqrt{n})$ copies of this construction and choose 
$k = \Theta(\sqrt{n})$.
The parts are independent in the sense that they are separated by 
\emph{unavoidable} edges (i.e., edges that are present in any 
triangulation of the resulting polygon).%
\footnote{Unavoidable edges are defined by segments between two vertices s.t.\ no other edge crosses them.
Hence, they have to be present in every triangulation.
Unavoidable edges of point sets have been investigated by Karoly and Welzl~\cite{karolyi_welzl} (as ``crossing-free segments''), and Xu~\cite{xu} (as ``stable segments'').}
Hence, while the dual diameter of a \MaxDT{} is linear in~$n$, it is 
in~$O(\sqrt{n})$ for any triangulation that minimizes the number of ears.
\end{proof}

\begin{figure}
\centering
\includegraphics[scale=0.8]{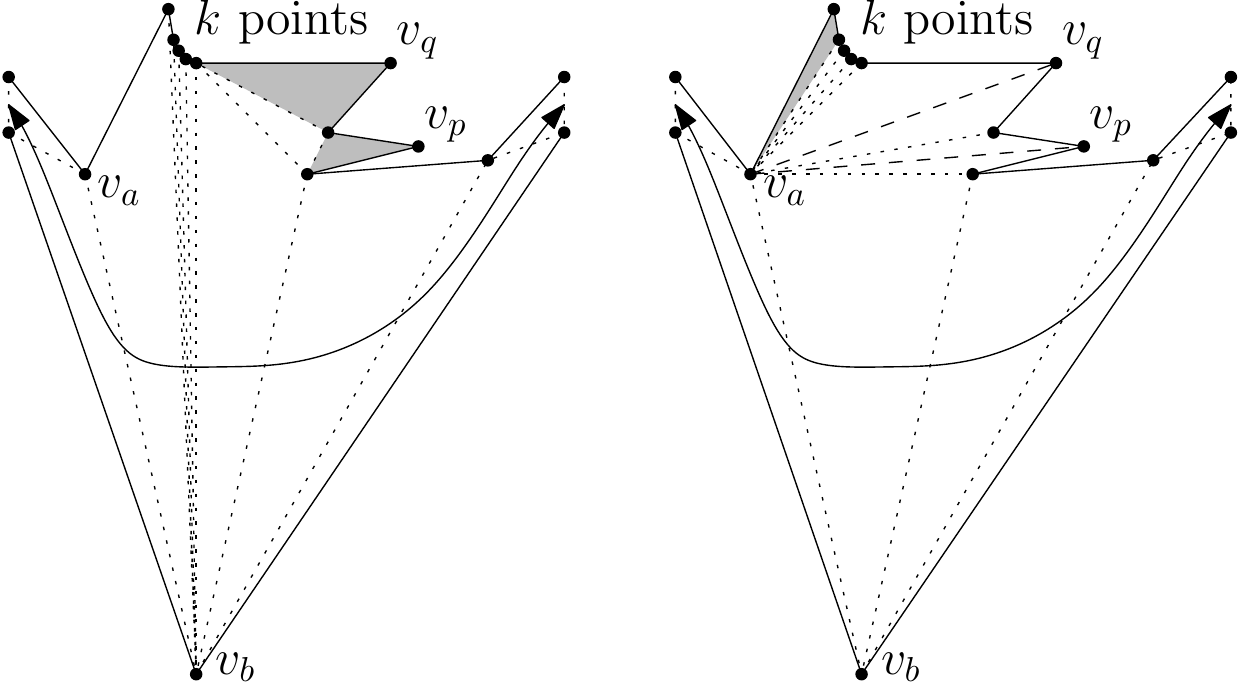}
\caption{Two triangulations of a part of a polygon where the dual diameter is locally decreased by~$k$ when minimizing the number of ears.}
\label{fig:fig_ears_max}
\end{figure}

It is easy to construct polygons for which the dual graph of any 
triangulation is a path, forcing minimum dual diameter $\Omega(n)$.
The other direction is slightly less obvious.

\begin{prop}
For any $n$, there exists a simple polygon~$\Poly$ with $n$ vertices such that the dual diameter of any \MaxDT{} of $\Poly$ is in~$\Theta(\log n)$.
\end{prop}
\begin{proof}
We incrementally construct $\Poly$ by starting with an arbitrary 
triangle~$t$.
See \figurename~\ref{fig:fig_max_log} for an accompanying illustration.
We replace every corner of~$t$ by four new vertices so that two of them 
 can see only these four new vertices.
This means that the edge between the other two newly added vertices is 
unavoidable.
We repeat this construction recursively in a balanced way.
If necessary, we add dummy vertices to 
obtain exactly~$n$ vertices.
The unavoidable edges partition $\Poly$ into convex regions, either 
hexagons or quadrilaterals. The dual tree of this partition is balanced 
with diameter $\Theta(\log n)$. 
Since every triangulation of $\Poly$ contains all unavoidable edges,
the maximum possible dual diameter is $O(\log n)$.
\end{proof}

\begin{figure}
\centering
\includegraphics[scale=0.8]{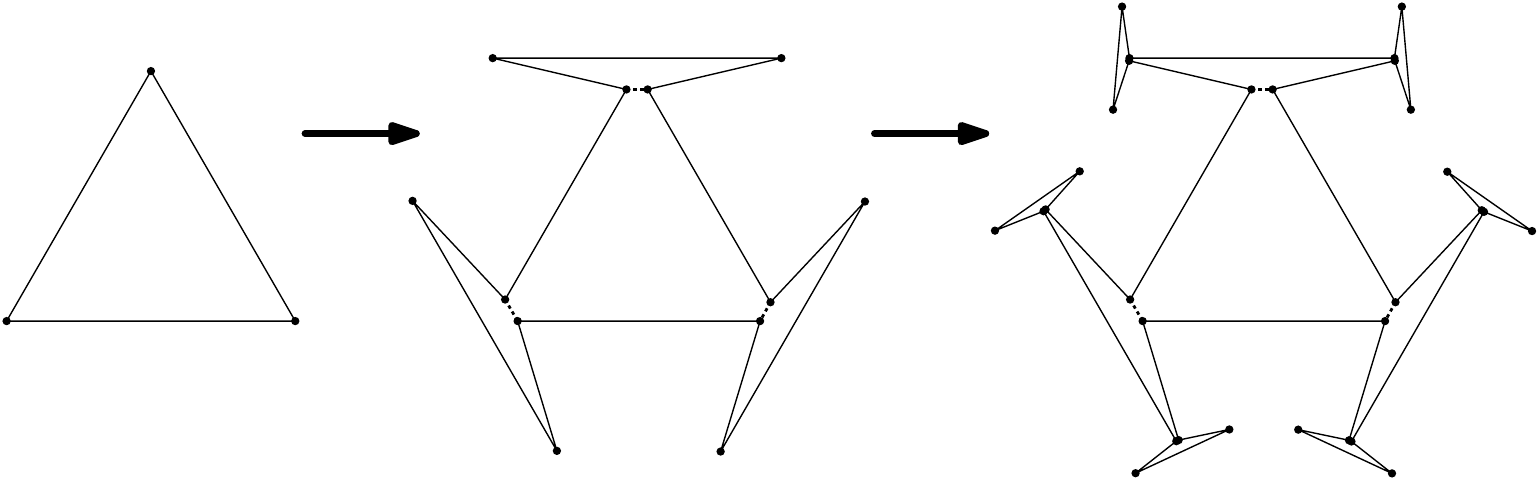}
\caption{The convex vertices of a polygon are incrementally replaced by four new vertices, resulting in unavoidable edges (dotted).}
\label{fig:fig_max_log}
\end{figure}

\section{Optimally Triangulating a Simple Polygon}\label{sec_poly}
We now consider the algorithmic question of constructing a $\MinDT$ 
and a $\MaxDT$ of a simple polygon $\Poly$ with $n$ vertices. 
Let $v_1,\ldots, v_n$ be the vertices of~$\Poly$ in counterclockwise order.
The segment $v_iv_j$ is a \emph{diagonal} of $\Poly$ if it lies 
completely in~$\Poly$ but is not part of
the boundary of~$\Poly$.
For a diagonal $v_iv_j$, $i < j$, we define $\Poly_{i,j}$ as the 
polygon with vertices $v_i, v_{i+1},\dots, v_{j-1}, v_j$; 
see \figurename~\ref{fig_pij}.
Observe that $\Poly_{i,j}$ is a simple polygon contained in~$\Poly$.
If $v_iv_j$ is not a diagonal, we set $\Poly_{i,j}=\emptyset$.

\begin{figure}[tb]
\centering
\includegraphics[width=.5\columnwidth]{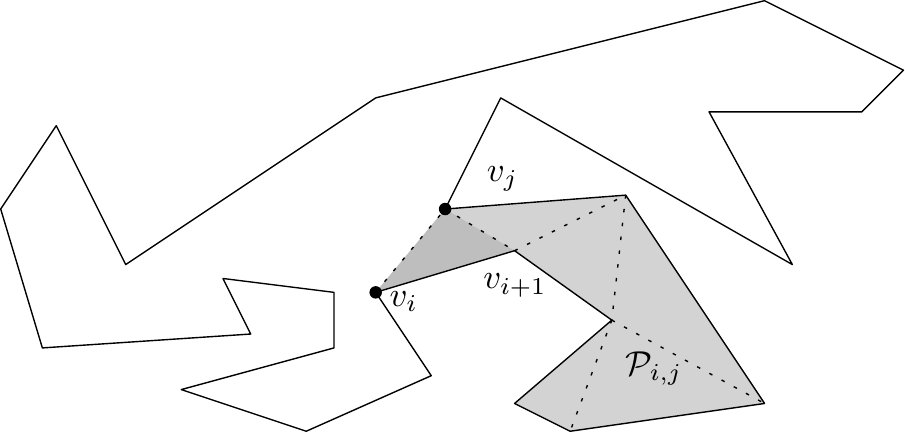}
\caption{Any triangulation of $\Poly_{i,j}$ 
(gray) has exactly one triangle adjacent to 
$v_iv_j$ (dark gray).}
\label{fig_pij}
\end{figure}

\begin{theorem}\label{thm_min_alg}
For any simple polygon $\Poly$ with $n$ vertices, we can compute a
$\MinDT$ in $O(n^3\log n)$ time.
\end{theorem}
\begin{proof}
We use the classic dynamic programming approach~\cite{klincsek}, with 
an additional twist to account for the non-local nature of the 
objective function.
Let $v_iv_j$ be a diagonal.
Any triangulation $\triang$ of $\Poly_{i,j}$ has exactly one 
triangle $t$ incident to $v_iv_j$; see \figurename~\ref{fig_pij}. 
Let $f(\triang)$ be the maximum length of a path in 
$\triang^*$ that has $t$ as an endpoint.

For $d>0$ and $i,j = 1, \dots, n$, with $i < j$, let 
$\mathcal{T}_d(i,j)$ be the set of all triangulations of $\Poly_{i,j}$ 
with dual diameter at most $d$ (we set 
$\mathcal{T}_d(i,j)=\emptyset$ if $v_iv_j$ is not a diagonal of $\Poly$).
We define $M_d[i,j]=\min_{\triang\in\mathcal{T}_d(i,j)}f(\triang)+1$, 
if $\mathcal{T}_d(i,j)\neq \emptyset$, or $M_d[i,j] = \infty$, otherwise.
Intuitively, we aim for a triangulation that minimizes the distance 
from $v_iv_j$ to all other triangles of $\Poly_{i,j}$ while keeping 
the dual diameter below $d$ (the value of $M_d[i,j]$ is the smallest 
possible distance that can be obtained). Let $\mathcal{V}(i,j)$ be 
all vertices $v_l$ of $\Poly_{i,j}$ such that the 
triangle $v_iv_jv_l$ lies inside $\Poly_{i,j}$.
We claim that $M_d[i,j]$ obeys the following recursion:

\[ M_d[i, j] = \left\{
\begin{array}{l}
  0, \hfill \mbox{if $i+1 = j$,}\\
  \infty, \hfill \mbox{if $v_iv_j$ is not a diagonal,}\\
  \min_{v_l\in\mathcal{V}(i,j)} I_d[i,j,l], \hfill \quad \mbox{otherwise,}
\end{array}
\right. 
\]
where
\[ I_d[i,j,l] = \left\{
\begin{array}{l}
  \infty,  \hfill \mbox{if $M_d[i,l] + M_d[l,j] > d$,} \\
  \max\{M_d[i,l], M_d[l,j]\}+1, \quad \mbox{otherwise.}
\end{array} \right.
\]

We minimize over all possible triangles $t$ in $\Poly_{i,j}$ incident to $v_iv_j$.
For each $t$, the longest path to $v_iv_j$ is the longer of the paths to the other edges of $t$ plus $t$ itself.
If $t$ joins two longest paths of total length more than $d$, there is no valid solution with $t$.
Thus, we can decide in $O(n^3)$ time whether there is a triangulation with dual diameter at most $d$, i.e., if $M_d[1,n] \neq \infty$.
Since the dual diameter is at most $n-3$, a binary search gives an $O(n^3 \log n)$ time algorithm.
\end{proof}

We can use a very similar approach to obtain some \MaxDT. 

\begin{theorem}
For any simple polygon $\Poly$ with $n$ vertices, we can compute a
$\MaxDT$ in $O(n^3\log n)$ time.
\end{theorem}
\begin{proof}
The proof is similar to the one of Theorem~\ref{thm_min_alg}.
This time, we are looking for triangulations that have dual diameter at 
least~$d$.
Let $f(T)$ be defined as before, and let $\mathcal{T}(i,j)$ be the set of 
all triangulations of $\Poly_{i,j}$ (this time, we do not need the 
third parameter).
We define $M_d[i,j]$ in the following way.
If $\mathcal{T}(i,j) = \emptyset$, then $M_d[i,j] = -\infty$.
If $\mathcal{T}(i,j)$ contains a triangulation with diameter at 
least~$d$, $M_d[i,j] = \infty$.
Otherwise, let $M_d[i,j] = \max_{T \in \mathcal{T}(i,j)} f(T) + 1$.
Clearly, there is a triangulation with diameter at 
least~$d$ if and only if $M_d[1,n] = \infty$.
With $\mathcal{V}(i,j)$ defined as before, the recursion for 
$M_d[i,j]$ is

\[ M_d[i, j] = \left\{
\begin{array}{l}
  0, \hfill \mbox{if $i+1 = j$,}\\
  -\infty, \hfill \mbox{if $v_iv_j$ is not a diagonal,}\\
  \max_{v_l\in\mathcal{V}(i,j)} I_d[i,j,l], \hfill \quad \mbox{otherwise,}
\end{array}
\right.
\]
where
\[ I_d[i,j,l] = \left\{
\begin{array}{l}
  -\infty, \hfill \mbox{if $M_d[i,l]$ or $M_d[l,j]$ is $-\infty$,} \\
  \infty,  \quad \hfill \mbox{if $M_d[i,l] + M_d[l,j] \geq d$,} \\
  \max\{M_d[i,l], M_d[l,j]\}+1, \quad \hfill \mbox{otherwise.}
\end{array} \right.
\]

For the given diagonal $v_i v_j$, we maximize over all possible triangles.
If at some point the triangle $t$ at $v_i v_j$ closes a path of length at 
least~$d$, we are basically done, as any triangulation of the remainder 
of the polygon results in a triangulation with dual diameter at least~$d$.
If the triangulation of $\Poly_{i,j}$ does not contain such a long path, 
we store the longer one to $v_i v_j$, as before.
Again, we can find the optimal dual diameter via a binary search, 
giving an $O(n^3 \log n)$ time algorithm.
\end{proof}

\section{Bounds for Point Sets}\label{sec_points}
We are now given a set $S$ of $n$ points
in the plane in general position, and we need to find a triangulation
of $S$ whose dual graph optimizes the diameter.
Since the dual graph has maximum degree $3$, it is easy to see that 
the $\Omega(\log n)$ lower bound for simple polygons extends 
for this case.
It turns out that this bound can always be achieved, as we show in 
Section~\ref{sec_min_point_set}.
In Section~\ref{sec_max_point_set}, we find a triangulation of~$S$ that 
has dual diameter in~$\Omega(\sqrt{n})$.

\subsection{Minimizing the Dual Diameter}\label{sec_min_point_set}

\begin{theorem}\label{theo_pointset}
Given a set $S$ of $n$ points in the plane in general position, we can
compute in $O(n \log n)$ time a triangulation
of $S$ with dual diameter $\Theta(\log n)$.
\end{theorem}

\begin{proof}
Let $\Poly$ be a convex polygon with $n$ vertices and
$\triang'$ a triangulation of $\Poly$ with dual diameter
$\Theta(\log n)$ (e.g., the triangulation from
Proposition~\ref{prop_convex}).
The triangulation $\triang'$ is an outerplanar graph.
Any outerplanar graph of $n$ vertices has a plane
straight-line embedding on any given $n$-point set~\cite{pgmp-eptvsp-91}.
Furthermore, such an embedding can be found in $O(nd)$
time and $O(n)$ space, where $d$ is the dual diameter of the
graph~\cite{b-oeopgps-02}.

Let $\triang_S$ be the embedding of $\triang'$ on~$S$.
In general, $\triang_S$ does not triangulate $S$; see \figurename~\ref{fig_pointset}.
\begin{figure}[tb]
\centering
\includegraphics[width=0.7\columnwidth]{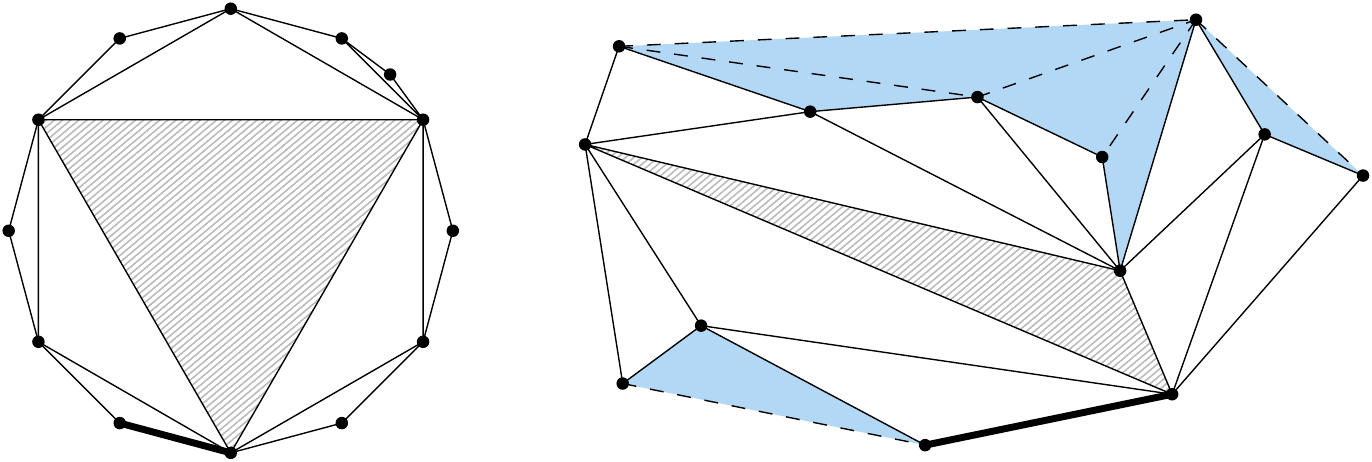}
\caption{When computing a $\MinDT$ of a point set $S$, we first view it as if it were in convex position and construct a $\MinDT$ (left image).
Then, we embed $T_S$ into the actual point set (solid edges in the right image).
(The correspondence is marked by the central triangle and the thick boundary edge.)
All remaining untriangulated pockets (highlighted region in the figure) are triangulated arbitrarily (dashed edges). }
\label{fig_pointset}
\end{figure}
Consider the convex hull of $\triang_S$ (which equals the convex hull of~$S$).
The untriangulated \emph{pockets} are simple polygons.
We triangulate each pocket arbitrarily to obtain
a triangulation $\triang$ of $S$.
%
We claim that the dual diameter of $\triang$ is $O(\log n)$.
%
\begin{lemma}\label{lem_logpoly}
The dual distance from any
triangle in a pocket to any triangle in $\triang_S$ 
is $O(\log n)$.
\end{lemma}
\begin{proof}
Let $Q$ be a pocket, and $\triang_Q$
a triangulation of $Q$.
Since $Q$ is a simple polygon, the dual $\triang_Q^*$ is a 
tree with maximum degree $3$.
A triangle $t$ of $\triang_Q$ 
not incident to the boundary of $\triang_S$ either has
degree $3$ in $\triang_Q^*$, or it is the unique triangle in $\triang_Q$ 
that shares an edge with the convex hull of $S$.
We perform a breadth-first-search in $\triang_Q^*$ starting from $t$,
and let $k$ be the maximum number of consecutive layers from the root
of the BFS-tree that
do not contain a triangle incident to the boundary. 
By the above observation, all vertices in the first $k-1$ levels have degree three in $\triang_Q^*$. Thus, each vertex of level $k-1$ or lower has two children. In particular, at each level the number of vertices must double (except at the topmost level where the number of vertices is tripled), hence $k = O(\log n)$. 
\end{proof}

Given Lemma~\ref{lem_logpoly} and the fact that
$\triang_S$ has dual diameter $O(\log n)$, 
Theorem~\ref{theo_pointset} is now immediate.
\end{proof}

\subsection{Obtaining a Large Dual Diameter}\label{sec_max_point_set}

We now focus our attention on the problem of triangulating $S$ so that the 
dual diameter is maximized. 

\begin{theorem}\label{theo_pointslargeDiam}
Given a set $S$ of $n \geq 3$ points in the plane in general position, we can
compute in $O(n \log n)$ time a triangulation
of $S$ with dual diameter at least $\sqrt {n-3}$.
\end{theorem}
\begin{proof}
Naturally, the triangulation $\triang$ must contain the edges of the convex hull of $S$.
Let $v_1,v_2,\ldots,v_h$ be the vertices of the convex hull of $S$ 
in clockwise order.
We connect $v_1$ to the vertices $v_3,v_4,\ldots,v_{h-1}$; see 
\figurename~\ref{fig:fig_points_max} (left).
In order to complete this set of edges to a triangulation, it suffices to consider 
the triangular regions $v_1v_iv_{i+1}$ (for $2\leq i \leq h-1$) 
with at least one point of $S$ in their interior.

\begin{figure}[bt]
\centering
\includegraphics{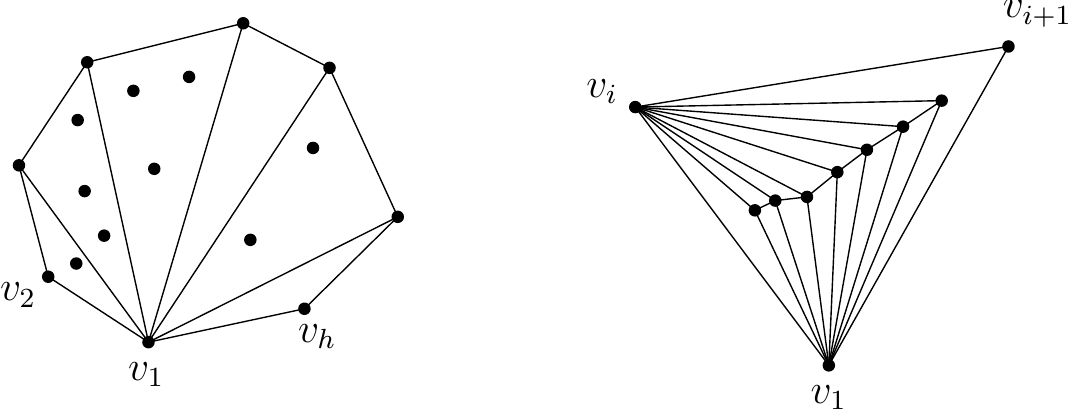}
\caption{Left: $v_1$ is connected to all remaining vertices in the convex hull.
Right: additional edges added inside $\Delta_i$. We connect the points of an increasing subsequence of $v_{j_1},v_{j_2},\ldots,v_{j_{n_i}}$ to both $v_1$, $v_i$ as well as the predecessor and successor in the subsequence.}
\label{fig:fig_points_max}
\end{figure}

Let $\Delta_i= v_1v_iv_{i+1}$ be such a triangular region, $S_i \subset S$ 
the points in the interior of $\Delta_i$, and $n_i=|S_i|$.
Let $w_1, w_2, \dots, w_{n_i}$ denote the points in $S_i$ sorted in 
clockwise order with respect to $v_1$, and $w_{j_1}, w_{j_2}, \dots,
w_{j_{n_i}}$ denote the same points sorted in counterclockwise order 
with respect to $v_i$.
By the Erd\H{o}s-Szekeres theorem~\cite{es_subsequence}, the index sequence $j_k$ 
contains an increasing or decreasing subsequence $\sigma_i$ of length 
at least $\sqrt{n_i}$.

If $\sigma_i$ is increasing, we connect all points of $\sigma_i$ to 
both $v_1$ and $v_i$. In addition, we connect each point of $\sigma_i$ 
to its predecessor and successor in $\sigma_i$; see 
\figurename~\ref{fig:fig_points_max} (right). Since the 
clockwise order with respect to $v_1$ coincides with 
the counterclockwise order with respect to $v_i$, the 
new edges do not create any crossing.

If $\sigma_i$ is decreasing, we claim that the corresponding point 
sequence is in counterclockwise order around $v_{i+1}$.
Indeed, let $w$ and $w'$ be two vertices of $S_i$ whose indices appear 
consecutively in $\sigma_i$ (with $w$ before $w'$). By definition, the 
segment $v_1w'$ crosses the segment $v_iw$. Moreover, points $v_1$ and 
$v_i$ are on the same side of the line through $w$ and~$w'$; see \figurename~\ref{fig:fig_points_subsequence} (left). Since 
$w$ and $w'$ are contained in $\Delta_i$, we conclude that $v_{i+1}$ 
must lie on the opposite side of the line. Thus, $v_i,w,w',v_1$ form a 
counterclockwise sequence around $v_{i+1}$, and
we can connect each point of $\sigma_i$ to $v_1$, $v_{i+1}$, and its 
predecessor and successor in $\sigma_i$ without crossings.
Finally, we add arbitrary edges to complete the resulting graph inside 
$\Delta_i$ into a triangulation $\triang_i$.
\begin{figure}[tb]
\centering
\includegraphics{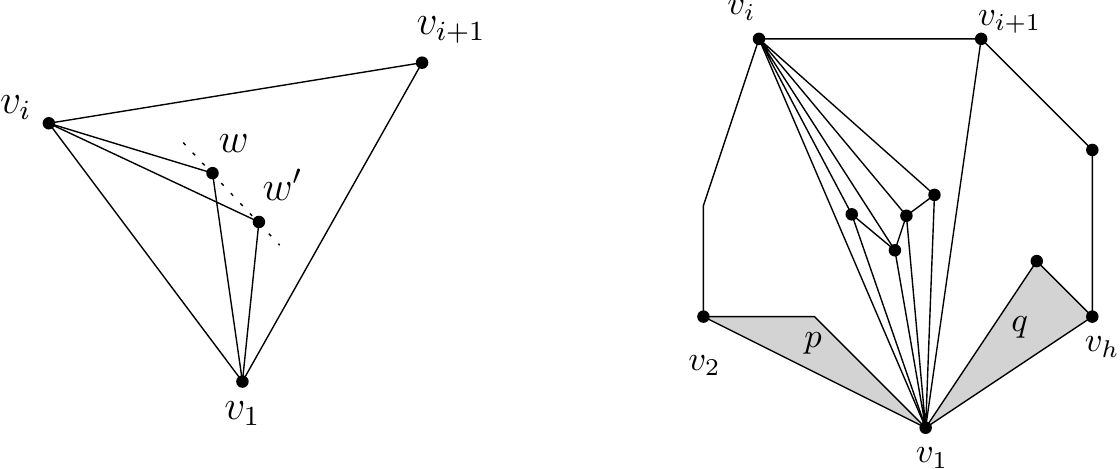}
\caption{Left: if $\sigma_i$ generates a decreasing sequence, the same sequence must be increasing when we view the angles with $v_{i+1}$ instead.
Right: any path between $p$ and $q$ in the dual graph must visit all triangles $\Delta_i$ and at least $\sqrt{n_i}$ additional triangles between the crossing of segments $v_1v_i$ and $v_1v_{i+1}$.}
\label{fig:fig_points_subsequence}
\end{figure}

 We claim that, regardless of how we complete the triangulation, there 
 are two triangles whose distance in the dual graph is at least 
 $\sqrt{n-3}$. 
Let $p$ and $q$ be the triangles of $\triang$ incident to edges $v_1v_2$ 
and $v_1v_{h}$, respectively (since both segments are on the convex hull, 
$p$ and $q$ are uniquely defined). Let $\pi$ be the shortest path from
$p$ to $q$. Clearly, $\pi$ must cross each segment $v_1v_i$, for 
$i\in \{3, \ldots h-1\}$, exactly once and in increasing order.
This gives $h-3$ steps (one step for each triangle 
incident in clockwise order around $v_1$ on an edge 
$v_1v_i$, $i \in \{3, \dots, h-1\}$).
In addition, 
at least $\sqrt{n_i}$ additional triangles must be traversed 
between the 
segments $v_1v_i$ and $v_1v_{i+1}$ (for all $i \in \{2, \ldots, h-1\}$): 
indeed, for each vertex $w \in \sigma_i$, the edges 
$v_1w$ and either $wv_i$ or $wv_{i+1}$ (depending on whether $\sigma_i$ 
was increasing or decreasing) disconnect $p$ and $q$.,
Hence at least one 
of the two must be crossed by $\pi$, and the triangles following these
edges are pairwise distinct and distinct from the triangles following
the segments $v_1v_i$.
Summing over $i$, we get
\[
|\pi| \geq 
h - 3 + \sum_{i=2}^{h-1}\sqrt{n_i} \geq 
h+\sqrt{n-h} - 3 \geq
3+\sqrt{n-3}-3 = 
\sqrt{n - 3}.
\]
In the second inequality we used the fact that 
$\sum_{i=2}^{h-1} \sqrt{n_i} \geq
\sqrt{\sum_{i=2}^{h-1} n_i} = \sqrt{n - h}$. (since a point is 
either on the convex hull or in its interior), the third inequality follows from $h\geq 3$ (and the fact that the expression is minimized when $h$ is as small as possible). 

Finding a longest increasing (or decreasing) subsequence of $n_i$ numbers 
takes $O(n_i \log n_i)$ time, which is optimal in the comparison 
model~\cite{fredman}.
Hence, all subsequences, as well as the whole triangulation can 
be computed in $O(n \log n)$ time,
where the last part also uses the fact that point location in a triangulation 
on $\leq n$ vertices can be done in $O(\log n)$ time after 
$O(n \log n)$ preprocessing.
\end{proof}

\begin{prop}\label{prp:convex_subset}
Any set of $n$ points in the plane in general position with $k$ points in convex position has a triangulation with dual diameter in $\Omega(k)$.
\end{prop}
\begin{proof}
See \figurename~\ref{fig:fig_convex_k_set} for an accompanying illustration.
Let $S$ be such a point set with $C \subseteq S$ being the convex subset of size~$k$.
First, triangulate only~$C$ by a zig-zag chain of edges: for the convex hull of~$C$ being defined by the sequence $(c_1, \dots, c_k)$, add the edges $c_i c_{k-i}$ and $c_i c_{k-i-1}$ for $1 \leq i < \lfloor k/2 \rfloor$, as well as the boundary of the convex hull of~$C$.
Then, add the extreme points of $S$ and triangulate the convex hull of $S$ without $C$ such that each added edge is incident to a point of $C$ (this is not necessary to obtain the result, but it makes our arguments simpler).
The resulting triangulation has dual diameter $\Omega(k)$,
as is witnessed by the triangles at $c_k$ and $c_{\left \lfloor k/2 \right \rfloor}$ inside $C$: if we label each triangle with the index of the incident point in $C$ that is closest to $c_{\left \lfloor{k/2} \right \rfloor}$, then this index can change by at most~1 along a step in any dual path
between the triangles inside $C$ incident to $c_k$ and
$c_{\left \lfloor k/2 \right \rfloor}$.
The dual diameter does not decrease when adding the remaining points of~$S$ and completing the triangulation arbitrarily.
\end{proof}

\begin{figure}
\centering
\includegraphics{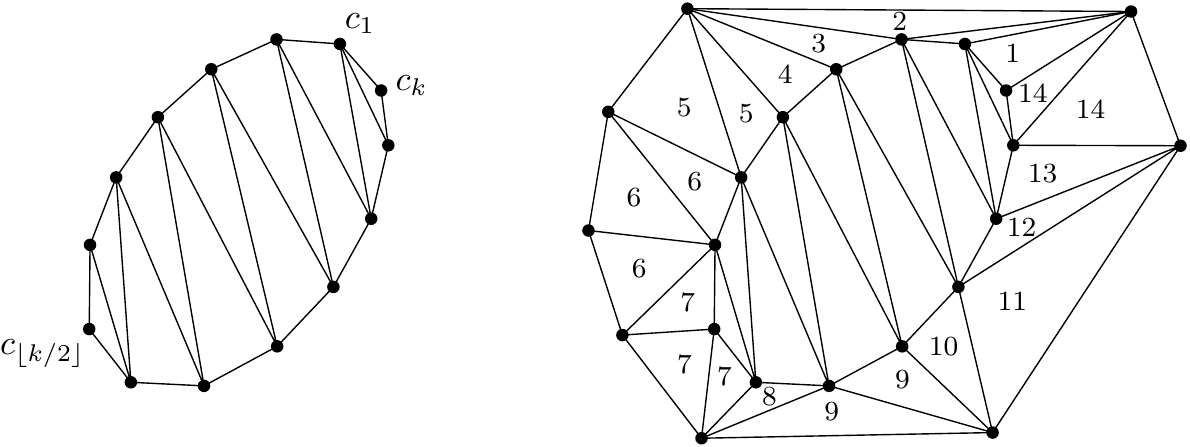}
\caption{Left: zig-zag triangulation of a convex subset of size~$k$.
Right: adding the edges to the extreme points implies a labeling of the triangles that relate to their distance to $c_{\left \lfloor k/2 \right \rfloor}$.}
\label{fig:fig_convex_k_set}
\end{figure}

\section{Conclusions}\label{sec_conclusion}

The proof of Corollary~\ref{cor_lb} (lower bound for simple polygons) 
is essentially based on fundamental properties of graphs (i.e., bounded 
degree) rather than geometric properties.
Since the bound is tight even for the convex case, it cannot be 
tightened in general.
However, we wonder if, by using geometric tools, one can construct a 
bound that depends on the number of reflex vertices of the polygon 
(or interior points for the case of sets of points).
Another natural open problem is to extend our dynamic programming 
approach for simple polygons to general polygonal domains 
(or even sets of points).

It is open whether Theorem~\ref{theo_pointslargeDiam} is tight. That is:
does there exist a point set~$S$ such that the diameter of the dual 
graph of any triangulation of $S$ is in $O(\sqrt{n})$? From Proposition~\ref{prp:convex_subset}, 
we see that any such point set can contain at most $O(\sqrt{n})$ points
in convex position. 
Thus, the point set must have $\Theta(\sqrt n)$ convex hull 
layers, each with $\Theta(\sqrt{n})$ points. We suspect that some smart 
perturbation of the grid may be an example, but we have been unable to 
prove so.  

\section*{Acknowledgements}
Research for this work was supported by the ESF EUROCORES
programme EuroGIGA--ComPoSe, Austrian
Science Fund (FWF): I 648-N18 and grant
EUI-EURC-2011-4306. 
M.~K. was supported in part by the ELC project (MEXT KAKENHI No. 24106008).
S.L. is Directeur de Recherches du FRS-FNRS.
W.M.\@ is supported in part by DFG grants
MU 3501/1 and MU 3501/2.
Part of this work has been done while A.P.\ was recipient of a 
DOC-fellowship
of the Austrian Academy of Sciences at
the Institute for Software Technology,
Graz University of Technology, Austria. 
M.S.\@ was supported by the project LO1506 of the Czech Ministry of Education, Youth and Sports, and by the project NEXLIZ CZ.1.07/2.3.00/30.0038, 
which was co-financed by the European Social Fund and the state budget of the Czech Republic.


\bibliographystyle{abbrv}
\bibliography{bibliography}

\newcommand{\SortNoop}[1]{}
\begin{thebibliography}{10}

\bibitem{b-oeopgps-02}
P.~Bose.
\newblock On embedding an outer-planar graph in a point set.
\newblock {\em Comput. Geom. Theory Appl.}, 23(3):303--312, 2002.

\bibitem{es_subsequence}
P.~Erd{\H{o}}s and G.~Szekeres.
\newblock A combinatorial croblem in geometry.
\newblock {\em Compositio Math.}, 2:463--470, 1935.

\bibitem{fredman}
M.~L. Fredman.
\newblock On computing the length of longest increasing subsequences.
\newblock {\em Discrete Math.}, 11(1):29--35, 1975.

\bibitem{hurtado_noy_urrutia}
F.~Hurtado, M.~Noy, and J.~Urrutia.
\newblock Parallel edge flipping.
\newblock In {\em Proc. 10th Canad. Conf. Comput. Geom. (CCCG)}, 1998.

\bibitem{karolyi_welzl}
G.~K{\'{a}}rolyi and E.~Welzl.
\newblock Crossing-free segments and triangles in point configurations.
\newblock {\em Discrete Applied Mathematics}, 115(1-3):77--88, 2001.

\bibitem{klincsek}
G.~T. Klincsek.
\newblock Minimal triangulations of polygonal domains.
\newblock {\em Ann. Discrete Math.}, 9:121--123, 1980.

\bibitem{kozma}
L.~Kozma.
\newblock Minimum average distance triangulations.
\newblock In {\em Proc. 20th Annu. European Sympos. Algorithms (ESA)}, pages
  695--706, 2012.

\bibitem{moore_survey}
M.~Miller and J.~\v{S}ir\'{a}\v{n}.
\newblock Moore graphs and beyond: A survey of the degree/diameter problem.
\newblock {\em Electronic J. of Combin.}, {\#}DS14, 2005.

\bibitem{pgmp-eptvsp-91}
J.~Pach, P.~Gritzmann, B.~Mohar, and R.~Pollack.
\newblock Embedding a planar triangulation with vertices at specified points.
\newblock {\em American {M}athematical {M}onthly}, 98:165--166, 1991.

\bibitem{s-cbtt-91}
T.~C. Shermer.
\newblock Computing bushy and thin triangulations.
\newblock {\em Comput. Geom. Theory Appl.}, 1:115--125, 1991.

\bibitem{sleator_tarjan_thurston}
D.~Sleator, R.~Tarjan, and W.~Thurston.
\newblock Rotation distance, triangulations and hyperbolic geometry.
\newblock {\em J. Amer. Math. Soc.}, 1:647--682, 1988.

\bibitem{urrutia_flip_talk}
J.~Urrutia.
\newblock Flipping edges in triangulations of point sets, polygons and maximal
  planar graphs.
\newblock Invited talk at the First Japan Conf. on Disc. and Comput. Geom.,
  1997.

\bibitem{xu}
Y.~Xu.
\newblock On stable line segments in all triangulations of a planar point set.
\newblock {\em Appl. Math. J. Chinese Univ. Ser. B}, 11(2):235--238, 1996.

\end{thebibliography}

\end{document}